\documentclass[a4paper,onecolumn,11pt,allowfontchangeintitle,accepted=2022-02-01]{quantumarticle}
\pdfoutput=1
\usepackage[utf8]{inputenc}
\usepackage[english]{babel}
\usepackage[T1]{fontenc}
\usepackage{amsmath}
\usepackage{hyperref}
\usepackage[numbers, sort&compress]{natbib}

\usepackage{tikz}
\usepackage{lipsum}
\usepackage{color, soul}

\usepackage{amsthm}
\usepackage{physics}
\usepackage{comment}
\usepackage{graphicx}
\usepackage{dcolumn}
\usepackage{bm}
\usepackage[section]{placeins}
\allowdisplaybreaks

\DeclareMathOperator{\Hil}{\mathcal{H}} 
\DeclareMathOperator{\B}{\mathcal{B}} 
\DeclareMathOperator{\St}{\mathcal{D}} 
\DeclareMathOperator{\N}{\mathcal{N}} 
\DeclareMathOperator{\M}{\mathcal{M}} 
\DeclareMathOperator{\Sp}{\mathcal{P}} 
\DeclareMathOperator{\R}{\mathcal{R}} 

\newcommand{\cvenn}[2]{\mathcal{F}^c(#1|#2)} 
\newcommand{\acvenn}[2]{\mathcal{F}^{ac}(#1|#2)} 
\newcommand{\seps}[2]{\mathcal{F}^s(#1|#2)} 

\DeclareMathOperator{\sep}{SEP}
\DeclareMathOperator{\uni}{UNI}
\DeclareMathOperator{\ncve}{NCVE}
\DeclareMathOperator{\eb}{EB}

\newcommand{\sepp}[2]{\sep(#1|#2)}
\newcommand{\ebp}[1]{\eb(#1)}
\newcommand{\unip}[2]{\uni(#1|#2)}

\newcommand{\ent}{S}
\newcommand{\cve}[2]{S_{#1|#2}}
\newcommand{\re}[2]{S \left(#1 \,\middle\| \,#2 \right)}

\newcommand{\cvennfree}[2]{\mathcal{O}^c_{max}(#1|#2)}
\newcommand{\cvenncompfree}[2]{\mathcal{O}^c_{cmax}(#1|#2)} 
  
\newcommand{\acvenncompfree}[2]{\mathcal{O}^{ac}_{cmax}(#1|#2)} 

\newcommand{\sepscompfree}[2]{\mathcal{O}^{s}_{cmax}(#1|#2)}
\newcommand{\ncvep}[2]{\ncve(#1|#2)} 

\newcommand{\cvennfreed}[4]{\mathcal{O}^c_{max}(#1|#2 \rightarrow #3|#4)}
\newcommand{\cvenncompfreed}[4]{\mathcal{O}^c_{cmax}(#1|#2 \rightarrow #3|#4)}

\newcommand{\sepsfreed}[4]{\mathcal{O}^{s}_{max}(#1|#2 \rightarrow #3|#4)}
\newcommand{\sepscompfreed}[4]{\mathcal{O}^{s}_{cmax}(#1|#2 \rightarrow #3|#4)}
\newcommand{\ncvepd}[4]{\ncve(#1|#2 \rightarrow #3|#4)}
\newcommand{\seppd}[4]{\sep(#1|#2 \rightarrow #3|#4)}
\newcommand{\ebpd}[2]{\eb(#1 \rightarrow #2)}

\DeclareMathOperator{\id}{id}

\newtheorem{theorem}{Theorem}

\newtheorem{lemma}{Lemma}

\newtheorem{definition}{Definition}

\newcommand*{\TitleFont}{%
      \fontsize{20}{20}%
      \selectfont}

\begin{document}

\title{\TitleFont A-unital Operations and Quantum Conditional Entropy}

\author{Mahathi Vempati}
\affiliation{Centre for Quantum Science and Technology,\\International Institute of Information Technology-Hyderabad, Gachibowli, Telangana-500032, India.}
\affiliation{Center for Computational Natural Sciences and Bioinformatics,\\International Institute of Information Technology-Hyderabad, Gachibowli, Telangana-500032, India.}
\email{mahathi.vempati@research.iiit.ac.in}

\author{Saumya Shah}
\affiliation{Department of Physics, \\ Indian Institute of Technology Guwahati, Guwahati-781039, Assam, India.}
\email{shah18@iitg.ac.in}

\author{Nirman Ganguly}
\affiliation{Department of Mathematics, \\ Birla Institute of Technology and Science Pilani, Hyderabad Campus, Telangana-500078, India.}
\email{nirmanganguly@gmail.com}

\author{Indranil Chakrabarty}
\affiliation{Centre for Quantum Science and Technology,\\International Institute of Information Technology-Hyderabad, Gachibowli, Telangana-500032, India.}
\affiliation{Center for Security, Theory and Algorithmic Research, \\International Institute of Information Technology-Hyderabad, Gachibowli, Telangana-500032, India.}
\email{indranil.chakrabarty@iiit.ac.in}
\
\maketitle

\begin{abstract}
Negative quantum conditional entropy states are key ingredients for information theoretic tasks such as superdense coding, state merging and one-way entanglement distillation. In this work, we ask: how does one detect if a channel is useful in preparing negative conditional entropy states? We answer this question by introducing the class of A-unital channels, which we show are the largest class of conditional entropy non-decreasing channels. We also prove that A-unital channels are precisely the completely free operations for the class of states with non-negative conditional entropy. Furthermore, we study the relationship between A-unital channels and other classes of channels pertinent to the resource theory of entanglement. We then prove similar results for ACVENN: a previously defined, relevant class of states and also relate the maximum and minimum conditional entropy of a state with its von Neumann entropy.

The definition of A-unital channels naturally lends itself to a procedure for determining membership of channels in this class. Thus, our work is valuable for the detection of resourceful channels in the context of conditional entropy. 
      
\end{abstract}


\section{Introduction}
The myriad information-theoretic applications of entanglement \cite{Horodecki2009} have elevated it from being merely an interesting physical phenomenon to that of a resource. In the resource-theoretic study of entanglement, the separable states (quantum states that can be written as convex combinations of pure product states) are called the free states, and those outside this set, the entangled states, are known as the resource states \cite{Chitambar2019}. The resource theory also defines various classes of free operations, that is, quantum operations that cannot convert a free state to a resourceful one. Local Operations and Classical Communication (LOCC) is a physically defined class of free operations \cite{Chitambar2014, Chitambar2019, Horodecki2009}, characterized by the practical difficulty of communicating quantum information.  A superset of LOCC is the class of separable channels \cite{Cirac2001, Chitambar2019}, which has a more appealing mathematical structure than LOCC and is therefore used to prove properties of LOCC itself \cite{Cirac2001, Vedral1998}. Entanglement-breaking channels form another important class of channels. These channels destroy all entanglement between a system and its environment \cite{Horodecki2003}. 

Once these classes of channels have been defined, it is important to be able to decide whether or not a given channel belongs to any of these free classes. The quantum channel detection problem asks: given a channel and a well-defined class of channels, does the channel belong to the class \cite{Macchiavello2013, Lee2020, Montanaro2016, Milazzo2020}? The channel itself maybe provided as a blackbox, as in a laboratory setting, where it can be characterized by measuring its output on various input states, or as a full description (for instance, via its Kraus decomposition or Choi matrix). Unfortunately, for some classes defined with respect to entanglement, this problem is notoriously difficult. For example, it is currently unknown how to decide whether a given channel belongs to LOCC \cite{Chitambar2019}, and the problem of membership of a channel in the entanglement-breaking class is NP Hard \cite{Gharibian2010}. There have been witness-based probabilistic approaches to deal with this problem, which are useful when some apriori information about a channel is known \cite{Macchiavello2013}. 

The interest in defining free operations for entanglement and studying mechanisms for detection of these operations is in part due to the resourcefulness of entanglement. While entanglement is certainly important in the characterization of quantum states for information-theoretic tasks, there are several phenomena for which merely possessing entanglement is not a sufficient condition to determine usefulness of a quantum state. Often, a state needs to have other properties beyond being entangled in order to be useful. One such property is the negativity of the quantum conditional entropy (also known as the conditional von Neumann entropy) of a quantum state \cite{Cerf1997}. The quantum conditional entropy (henceforth referred to as conditional entropy where there is no ambiguity) of a quantum state $\rho_{AB}$ is defined as   
\begin{equation}
    \cve{A}{B}(\rho_{AB}) = \ent(\rho_{AB}) - \ent(\rho_B),  
    \label{cve-def}
\end{equation}
where $\ent(\rho)$ is the von Neumann entropy (henceforth referred to as entropy where there is no ambiguity) of the state $\rho$, given by $\ent(\rho) = -\Tr(\rho \log \rho)$. All logarithms are taken to the base 2. Conditional entropy is used to define the information-theoretic quantities \emph{coherent information} \mbox{\cite{Wilde2013}}, which is just the negative of conditional entropy, and \emph{quantum discord} \mbox{\cite{Ollivier2001, Henderson2001, Radhakrishnan2020}}, which is a measure of non-classical correlations in a quantum state. While all negative conditional entropy states are entangled, not all entangled states possess negative conditional entropy. Negative conditional entropy states are important for the information-theoretic primitive quantum state merging \cite{Horodecki2005, Horodecki2006}, and provide quantum advantage in superdense coding \cite{Bennett1992, Bruss2004, Prabhu2013}. Negativity of conditional entropy is a necessary feature for the one-way distillation of entanglement \cite{Devetak2005}. It is important in maximizing the rates of distributed private randomness distillation \cite{Yang2019}, as well as in the reduction of uncertainty in predicting the outcomes of two incompatible measurements \cite{Berta2010}. Previously, the set of states that are not useful for any of the above tasks, that is, the set of states possessing non-negative conditional entropy (we refer to this set as $\cvenn AB$) was characterized as convex and compact \cite{Vempati2021, Friis2017}.\footnote{This is proven for the the case where $\dim A = \dim B = d$ \cite{Vempati2021}, but a similar argument can be used to extend this proof to the case where $\dim A \neq \dim B$.}The set of states for which conditional entropy remains non-negative even under the application of global unitary channels (referred to as $\acvenn AB$ in this work) was also similarly characterized \cite{Patro2017}.  

In this work, taking $\cvenn AB$ to be a set of free states, we ask: what is the class of free operations for $\cvenn AB$, and what are its properties? How does this class of operations relate to operations defined in the context of entanglement: separable channels and entanglement-breaking channels? How does one detect if a given channel belongs to the class of free operations for $\cvenn AB$? These questions are important in the context of preparation of quantum states for the aforementioned tasks, in order to determine if an operation one has access to is capable of generating resource. 

This paper is organized as follows: In Section \ref{sec:notation}, we introduce the notation, and recap the relevant preliminaries. In Section \ref{sec:freeops}, we define a class of free operations $\ncvep AB$, the largest class of conditional entropy non-decreasing channels. We prove that this class is exactly equal to $\cvenncompfree AB$, the largest class of completely free channels for $\cvenn AB$, and is elegantly characterized by a property we call $A$-unitality. We discuss properties of $A$-unital channels in Section \ref{sec:props}. In Section \ref{sec:aunital_sep}, we discuss the relation between $A$-unital channels, separable channels and entanglement-breaking channels. While the set of separable states and $\cvenn AB$ share a subset-superset relationship, we show that surprisingly, this relationship does not hold for $A$-unital channels and separable channels---there exist $A$-unital channels that are not separable, separable channels that are not $A$-unital, as well as channels in the intersection. In Section \ref{sec:unital}, we extend the class $\acvenn AB$, previously characterized only for the $2 \otimes 2$ case to the $d \otimes d$ case, and prove that unital channels are the largest class of completely free channels for $\acvenn AB$. Since both entropy and conditional entropy are heavily discussed in this manuscript, in Section \ref{sec:bounds}, we provide upper and lower bounds for the conditional entropy of a state for a given entropy. We conclude in Section \ref{sec:conclusion} with some concrete open questions that emerge from our work.

\section{Preliminaries}
\label{sec:notation}

In this section, we review relevant definitions and concepts from literature. $\Hil^A$ denotes a finite dimensional Hilbert space for physical system $A$. The dimension of the space $\Hil^A$ is given by $d_A$. The set of bounded linear operators on the Hilbert space $\Hil^A$ is denoted by $\B(A)$. The set of bounded linear operators over Hilbert space $\Hil^A \otimes \Hil^B$ is $\B(AB)$. The set of bounded linear operators that are Hermitian, positive semi-definite and have unit trace, that is, the density matrices (or quantum states) over $\Hil^A$ is denoted by $\St(A)$. A channel $\N: \St(A) \rightarrow \St(B)$ is a linear map from $\St(A)$ to $\St(B)$ that is completely positive and trace preserving (CPTP). 

A quantum state $\rho_{AB}$ is called separable if it can be written as a convex combination of pure product states. That is, $\rho_{AB}$ is separable if $\rho_{AB} = \sum_i p_i |\alpha_i\rangle \langle \alpha_i|$
where $p_i \in [0, 1]$, $\sum_i p_i = 1$ and $|\alpha_i\rangle = |\phi_i\rangle_A \otimes |\psi_i\rangle_B$ for some $|\phi_i\rangle \in \Hil^A$ and $|\psi_i\rangle \in \Hil^B$. Separable states are the free states in the resource theory of entanglement. We denote the set of separable states in $\St(AB)$ by $\seps AB$. The set of states $\rho_{AB} \in \St(AB)$ for which $\cve AB(\rho_{AB}) \geq 0$, that is, the set of states possessing non-negative conditional entropy (Eq.~\mbox{\ref{cve-def}}) is denoted by $\cvenn AB$. 

Free and completely free channels for a given set of free states were defined in \cite{Chitambar2019}. We define these here for $\cvenn AB$. 

\begin{definition}
    \label{def:cvennfree}
    A channel $\N_{AB \rightarrow CD} : \St(AB) \rightarrow \St(CD)$ is free for $\cvenn AB$ if for all $ \rho \in \cvenn AB$, we have  $\N_{AB \rightarrow CD}(\rho) \in \cvenn CD$. The class of all such free channels mapping $\St(AB)$ to $\St(CD)$ is denoted by $\cvennfreed ABCD$. When the input and output spaces are the same, the class is denoted by $\cvennfree AB \equiv \cvennfreed ABAB$. 
\end{definition}

\begin{definition}
    \label{def:cvenncompfree}
    A channel $\N_{AB \rightarrow CD} : \St(AB) \rightarrow \St(CD)$ is completely free for $\cvenn AB$ if for all finite dimensional Hilbert spaces $\Hil^{A'}$  and $\Hil^{B'}$ and $\rho \in \cvenn {A'A}{B'B}$, we have  $\left( \id_{A'B'} \otimes \N_{AB \rightarrow CD} \right)(\rho) \in \cvenn {A'C}{B'D}$. The class of all such completely free channels mapping $\St(AB)$ to $\St(CD)$ is denoted by $\cvenncompfreed ABCD$. When the input and output spaces are the same, the class is denoted by $\cvenncompfree AB \equiv \cvenncompfreed ABAB$. 
\end{definition}

The free and completely free channels for $\seps AB$, denoted by $\sepsfreed ABCD$ and $\sepscompfreed ABCD$ respectively are defined analogously, replacing $\cvenn AB$ in the above definitions by $\seps AB$. 

Separable channels \cite{Cirac2001} form an important class of channels in the entanglement resource theory. We define these channels here:

\begin{definition}
    \label{def:sep}
    A channel $N: \St(AB) \rightarrow \St(CD)$ is called separable if and only if its action on a state $\rho \in \St(AB)$ can be expressed as 
    \begin{equation}
        \N(\rho) = \sum_{i=1}^n (K^i_A \otimes L^i_B) \rho (K^i_A \otimes L^i_B)^\dagger
    \end{equation}
    for some integer $n \geq 1$ and linear maps $K^i_A:\Hil^A \rightarrow \Hil^C$ and $L^i_B:\Hil^B \rightarrow \Hil^D$ such that $\sum_{i=1}^n (K^i_A \otimes L^i_B)^\dagger (K^i_A \otimes L^i_B) = I_{AB}$. The class of all separable channels mapping $\St(AB)$ to $\St(CD)$ is denoted by $\seppd ABCD$. When the input and output spaces are the same, the class is denoted by $\sepp AB \equiv \seppd ABAB$.
\end{definition}
It turns out that $\sepp AB$ is exactly equal to $\sepscompfree AB$ \cite{Chitambar2019}. That is, the class of separable channels form the largest class of completely free channels for $\seps AB$. 

Entanglement-breaking channels constitute another relevant class of channels \cite{Wilde2013}:

\begin{definition}
    \label{def:eb}
    A channel $\N: \St(A) \rightarrow \St(B)$ is called entanglement-breaking if the quantum state $(\id_R \otimes \N) (\rho_{RA})$ is separable for all $\rho_{RA}$ where $R$ is a reference system of arbitrary dimension. The class of all entanglement-breaking channels mapping $\St(A)$ to $\St(B)$ is denoted by $\ebpd AB$. When the input and output spaces are the same, the class is denoted by $\ebp A \equiv \ebpd AA$.
\end{definition}

Unital channels play an important role in the context of entropy, as the class of unital channels mapping $\St(A)$ to itself are exactly the ones that do not decrease the entropy of any state $\rho \in \St(A)$ \cite{Wilde2013}. We define them as

\begin{definition}
    \label{def:unital}
    A channel $\N: \St(A) \rightarrow \St(A)$ is called unital if $\N \left(\frac{I_A}{d_A}\right) = \frac{I_A}{d_A}$. The class of all unital channels mapping $\St(A)$ to $\St(A)$ is denoted by $\uni(A)$. 
\end{definition}

Having reviewed the definitions, we now characterize the class of completely free operations for $\cvenn AB$.

\section{Completely Free Operations for Conditional Entropy}
\label{sec:freeops}

Since the negativity of conditional entropy implies usefulness for various tasks, a natural candidate for a class of operations that are free for $\cvenn AB$ is the class of channels that does not decrease the conditional entropy of any state. Conditional entropy non-decreasing channels are defined as follows:

\begin{definition}
    \label{def:ncve}
    A channel $\N_{AB \rightarrow CD}: \St(AB) \rightarrow \St(CD)$ is said to be conditional von Neumann entropy (CVE) non-decreasing, if for all $\rho \in \St(AB)$, we have $\cve CD \left( \N_{AB \rightarrow CD}(\rho) \right) \geq \cve AB(\rho)$. The class of CVE non-decreasing channels mapping $\St(AB)$ to $\St(CD)$ is denoted by $\ncvepd ABCD$. When the input and output spaces are the same, the class is denoted by $\ncvep AB \equiv \ncvepd ABAB$. 
\end{definition}

It turns out that this class of channels is exactly equal to $\cvenncompfree AB$ which was introduced in Def.~\ref{def:cvenncompfree}. We prove this equivalence by showing that both $\ncvep AB$ and $\cvenncompfree AB$ are in fact equal to a third class of operations, defined as follows:

\begin{definition}
    \label{def:auni}
    A channel $\N: \St(AB) \rightarrow \St(AB)$ is called A-unital if for all $\rho_B \in \St(B)$, $\N$ satisfies 
    \begin{equation}
        \N \left( \frac{I_A}{d_A} \otimes \rho_B \right) = \frac{I_A}{d_A} \otimes \rho'_B
    \end{equation}
    for some $\rho'_B \in \St(B)$. We denote the set of A-unital channels on $\St(AB)$ as $\unip AB$. Likewise, the channel $\N$ is called $B$-unital if for all $\rho_A \in \St(A)$, $\N$ satisfies
     \begin{equation}
        \N \left( \rho_A \otimes \frac{I_B}{d_B}  \right) = \rho'_A \otimes \frac{I_B}{d_B}
    \end{equation}
    for some $\rho'_A \in \St(A)$, and the set of $B$-unital channels on $\St(AB)$ is denoted as $\unip BA$.
\end{definition}

Not only does the class of $A$-unital channels serve to show that both $\ncvep AB$ and $\cvenncompfree AB$ are equal, but its definition also makes this class amenable to detection, as discussed in Section \ref{subsec:detection}. We now prove the equivalence of these three classes. 

\begin{theorem}
    \label{thm:main_theorem}
    The classes $\unip AB$, $\ncvep AB$ and $\cvenncompfree AB$ are all equal to each other.
\end{theorem}

\begin{figure}
    \includegraphics[width=\linewidth]{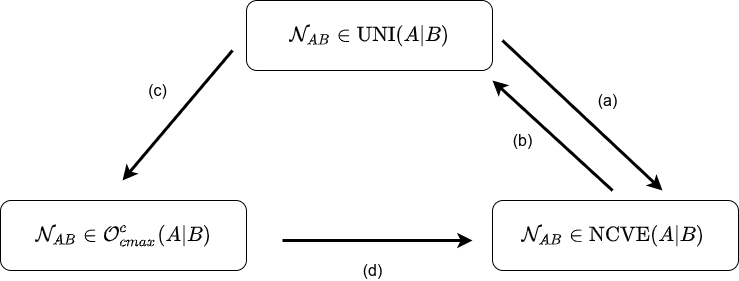}
    \caption{Equivalence of $\unip AB, \ncvep AB$ and $\cvenncompfree AB$.}
    \label{fig:a_uni_proof}
\end{figure}

\begin{proof}
    The structure of the proof is depicted in Fig.~\ref{fig:a_uni_proof}. Let $\N:\St(AB) \rightarrow \St(AB)$ be a quantum channel. In Lemma \ref{lem:auni_ncve}, we show that $\N \in \unip AB \implies \N \in \ncvep AB$, labelled as (a). In Lemma \ref{lem:ncve_auni}, we show that $\N \in \ncvep AB \implies \N \in \unip AB$, which is labelled as (b). In Lemma \ref{lem:auni_aauni}, we show that $\N \in \unip AB \implies (\id_{A'B'} \otimes \N) \in \unip {A'A}{B'B}$. We then use Lemma \ref{lem:auni_ncve} and Lemma \ref{lem:auni_aauni} to prove Lemma \ref{lem:auni_ocmax}, which is $\N \in \unip AB \implies \N \in \cvenncompfree AB$, labelled as (c) in the figure. Finally, we prove Lemma \ref{lem:ocmax_ncve}, which is $\N \in \cvenncompfree AB \implies \N \in \ncvep AB$, labelled as (d) in the figure. These lemmas establish the equivalence between the three classes.

\end{proof}

\begin{lemma}
    \label{lem:auni_ncve}
    $\N \in \unip AB \implies \N \in \ncvep AB$.
\end{lemma}

\begin{proof}
    Let $\N \in \unip AB$. Thus, for all $\rho_B \in \St(B)$, we have $\N \left(\frac{I_A}{d_A} \otimes \rho_B\right) = \frac{I_A}{d_A} \otimes \rho'_B$. To prove the Lemma, we require two results related to quantum relative entropy. The first result is a characterization of the conditional entropy in terms of relative entropy \cite{Nielsen2000, Wilde2013}:  
    \begin{align}
    \begin{split}
        \label{eq:cve_rel_entropy}
        \cve AB (\rho_{AB}) &= \log d -\re{\rho_{AB}} {\frac{I_A}{d_A} \otimes \rho_B} \\
                            &= \underset{\sigma_B \in \St(B)}{\max} \log d -\re {\rho_{AB}} {\frac{I_A}{d_A} \otimes \sigma_B}.  
    \end{split}
    \end{align}
    Next, we require the monotonicity of relative entropy, which is
    \begin{align}
        \label{eq:mono_rel_entropy}
        \re {\N(\rho)} {N(\sigma)} \leq \re \rho \sigma.
    \end{align}
    From Eq.~\ref{eq:cve_rel_entropy}, for all $\rho \in \St(AB)$ and $\sigma_B \in \St(B)$, we have 
    \begin{equation}
        \cve AB \left(\N(\rho)\right) \geq \log d -\re {\N(\rho)} {\frac{I_A}{d_A} \otimes \sigma_B}.
    \end{equation}
    In particular, we must have 
    \begin{equation}
        \cve AB \left(\N(\rho)\right) \geq \log d-\re {\N(\rho)} {\frac{I_A}{d_A} \otimes \rho'_B},
    \end{equation}where $\N \left(\frac{I_A}{d_A} \otimes \rho_B \right) = \frac{I_A}{d_A} \otimes \rho'_B$. 
    Therefore, 
    \begin{equation}
        \cve AB \left(\N(\rho)\right) \geq \log d-\re {\N(\rho)}{\N \left(\frac{I_A}{d_A} \otimes \rho_B\right)}.
    \end{equation}
    From Eq.~\ref{eq:mono_rel_entropy}, we have 
    \begin{equation}
        \log d -\re {\N(\rho)} {\N \left(\frac{I_A}{d_A} \otimes \rho_B \right)} \geq \log d -\re {\rho} {\frac{I_A}{d_A} \otimes \rho_B} = \cve AB (\rho).
    \end{equation} 
    Thus, for all $\rho \in \St(AB)$, we have $\cve AB \left(\N(\rho)\right) \geq \cve AB (\rho)$, which implies that $\N \in \ncvep AB$.

\end{proof}

\begin{lemma}
    \label{lem:ncve_auni}
    $\N \in \ncvep AB \implies \N \in \unip AB$.
\end{lemma}

\begin{proof}
    For any $\rho_{AB} \in \St(AB)$, the maximum value of $\cve AB (\rho_{AB})$ is $\log d_A$ and is attained if and only if $\rho_{AB} = \frac{I_A}{d_A} \otimes \sigma_B$ for some $\sigma_B \in \St(B)$ \cite{Wilde2013}. 

    \sloppy We use the above fact to prove the lemma by contradiction. Let $\N \in \ncvep AB$. Assume that $\N \notin \unip AB$. Then, there exists some $\sigma_B \in \St(B)$ such that $\N \left(\frac{I_A}{d_A} \otimes \sigma_B \right) \neq \frac{I_A}{d_A} \otimes \sigma_B'$ for all $\sigma_B' \in \St(B)$. This implies that $\cve AB \left(\frac{I_A}{d_A} \otimes \sigma_B \right) = \log d_A$, but $\cve AB \left(\N \left(\frac{I_A}{d_A} \otimes \sigma_B \right)\right) < \log d_A$. However, this is a contradiction since $\N \in \ncvep AB$. Therefore, the assumption that $\N \notin \unip AB$ must be false, and we have $\N \in \ncvep AB \implies \N \in \unip AB$. 
\end{proof}

\begin{lemma}
    \label{lem:auni_aauni}
$\N \in \unip AB \implies (\id_{A'B'} \otimes \N) \in \unip {A'A}{B'B}$.
\end{lemma}

\begin{proof}
    Let $\N_{AB} \in \unip AB$. Then, for all $\rho_B \in \St(B)$, we have $\N \left(I_A \otimes \rho_B \right) = I_A \otimes \rho'_B$ for some $\rho'_B \in \St(B)$. It is always possible to find a basis for any $\St(X)$, in which each basis element is a density matrix, that is, each basis element also belongs to $\St(X)$ \cite{Wilde2013}. Let $\{\beta'_i\}_i$ be such a basis for $\St(B')$ and $\{\beta_j\}_j$ be such a  basis for $\St(B)$. Thus, $\{\beta'_i \otimes \beta_j\}_{ij}$ is a basis of density matrices for $\St(B'B)$. We have,
    \begin{align}
        \begin{split}
            (\id_{A'B'} & \otimes \N_{AB}) \left(\frac{I_{A'A}}{d_{A'A}} \otimes \left(\beta'_i \otimes \beta_j \right)_{B'B}\right)  \\   
                        &= (\id_{A'B'} \otimes \N_{AB}) \left(\frac{I_{A'}}{d_{A'}} \otimes \frac{I_A}{d_A} \otimes \beta'_{iB'} \otimes \beta_{jB} \right) \\  
                        &= \frac{I_{A'}}{d_{A'}} \otimes \beta'_{iB'} \otimes \N_{AB} \left( \frac{I_A}{d_A} \otimes \beta_{jB} \right) \\  
                        &= \frac{I_{A'}}{d_{A'}} \otimes \beta'_{iB'} \otimes \frac{I_A}{d_A} \otimes \gamma_{jB} \\  
                        &= \frac{I_{A'A}}{d_{A'A}} \otimes \left(\beta'_{i} \otimes \gamma_j\right)_{B'B}.  
        \end{split}
    \end{align}

    Now, for all $\sigma \in \St(B'B)$, we have  
    \begin{align}
        \begin{split}
            (\id_{A'B'} & \otimes \N_{AB}) \left(\frac{I_{A'A}}{d_{A'A}} \otimes \sigma_{B'B}\right)  \\   
                        &=(\id_{A'B'}  \otimes \N_{AB}) \left(\frac{I_{A'A}}{d_{A'A}} \otimes \left(\sum_{ij}  b_{ij}\beta'_i \otimes \beta_j\right)_{B'B} \right)  \\   
                        &=\sum_{ij} b_{ij} \, (\id_{A'B'}  \otimes \N_{AB}) \left(\frac{I_{A'A}}{d_{A'A}} \otimes \left(\beta'_i \otimes \beta_j\right)_{B'B} \right)  \\   
                        &= \sum_{ij} b_{ij} \, \frac{I_{A'A}}{d_{A'A}} \otimes \left(\beta'_{i} \otimes \gamma_j\right)_{B'B}\\ 
                        &= \frac{I_{A'A}}{d_{A'A}} \otimes \sum_{ij} b_{ij} \left(\beta'_{i} \otimes \gamma_j\right)_{B'B}.\\ 
        \end{split}
    \end{align}

    Thus, $\id_{A'B'} \otimes \N_{AB} \in \unip {A'A}{B'B}$.
\end{proof}

\begin{lemma}
    \label{lem:auni_ocmax}
    $\N \in \unip AB \implies \N \in \cvenncompfree AB$.
\end{lemma}

\begin{proof}
    \sloppy Let $\N_{AB} \in \unip AB$. From Lemma \ref{lem:auni_aauni}, we have $\id_{A'B'} \otimes \N_{AB} \in \unip {A'A}{B'B}$. Hence, using Lemma \ref{lem:auni_ncve}, we have $\id_{A'B'} \otimes \N_{AB} \in \ncvep {A'A}{B'B}$. Thus, for all $\rho \in \cvenn {A'A}{B'B}$, we have $\cve AB \left(\left(\id_{A'B'} \otimes N_{AB}\right) (\rho) \right) \geq \cve AB (\rho) \geq 0$, which implies $\left(\id_{A'B'} \otimes \N_{AB}\right)(\rho) \in \cvenn {A'A}{B'B}$. Now, using Definition \ref{def:cvenncompfree}, we have $\N_{AB} \in \cvenncompfree AB$. 
\end{proof}

\begin{lemma}
    \label{lem:ocmax_ncve}
    $\N \in \cvenncompfree AB \implies \N \in \ncvep AB$.
\end{lemma}

\begin{proof}
    In order to prove the lemma, we prove the contrapositive statement, that is $\N_{AB} \notin \ncvep AB \implies \N_{AB} \notin \cvenncompfree AB$.
    Let $\N_{AB}:\St(AB) \rightarrow \St(AB)$ be a quantum channel that does not belong to $\ncvep AB$. Thus there exists $\rho \in \St(AB)$ such that $\cve AB \left( \N(\rho) \right) < \cve AB (\rho)$. 
    Consider Hilbert spaces $\Hil^{A'}$ and $\Hil^{B'}$, with $d_{A'} = d_A$ and $d_{B'} = d_B$. Now, we select a state $\sigma \in \St(A'B')$ such that   
    \begin{equation}
        \label{eq:ncve_ineq}
        - \cve AB (\rho) \leq \cve {A'}{B'} (\sigma) < - \cve AB \left(\N(\rho) \right).
    \end{equation}
    \sloppy Note that because $- \cve AB (\rho) < - \cve AB \left(\N(\rho)\right)$ and both these values lie in $[- \log d_A, \log d_A] = [- \log d_{A'}, \log d_{A'}]$, the state $\sigma$ always exists.  

        Now, consider the state $\sigma_{A'B'} \otimes \rho_{AB} \in \St(A'B'AB)$. Using Eq.~\ref{eq:ncve_ineq}, we have $\cve {A'A}{B'B} (\sigma_{A'B'} \otimes \rho_{AB}) = \cve {A'}{B'}(\sigma) + \cve AB (\rho) \geq - \cve AB (\rho) + \cve AB(\rho) = 0$. Thus, $\sigma_{A'B'} \otimes \rho_{AB} \in \cvenn{A'A}{B'B}$.

        We now apply the operation $\id_{A'B'} \otimes \N_{AB}$ on the state $\sigma_{A'B'} \otimes \rho_{AB}$. Once again, using Eq.~\ref{eq:ncve_ineq}, we have $\cve {A'A}{B'B} (\id_{A'B'} \otimes \N_{AB} (\sigma_{A'B'} \otimes \rho_{AB} )) = \cve {A'A}{B'B} (\sigma_{A'B'} \otimes \N_{AB}(\rho_{AB})) $ $= \cve {A'}{B'} (\sigma_{AB}) + \cve {A}{B} (\N_{AB}(\rho_{AB})) < - \cve AB (\N_{AB}(\rho_{AB})) + \cve AB (\N_{AB}(\rho_{AB})) = 0$. This implies that $\id_{A'B'} \otimes \N_{AB} (\sigma_{A'B'} \otimes \rho_{AB} ) \notin \cvenn {A'A}{B'B}$. Therefore, the channel $\id_{A'B'} \otimes \N_{AB}$ takes a state from inside $\cvenn {A'A}{B'B}$ to outside, and using Def.~\ref{def:cvenncompfree}, $\N_{AB} \notin \cvenncompfree AB$.
\end{proof}
Thus, Lemmas \ref{lem:auni_ncve}-\ref{lem:ocmax_ncve} prove all the implications depicted in Fig.~\ref{fig:a_uni_proof}, and Theorem \ref{thm:main_theorem} follows.

\section{Properties of A-Unital channels}
\label{sec:props}
\subsection{Serial and Parallel Concatenation}
Since we have $\unip AB = \cvenncompfree AB$, it follows that serial concatenation of two $A$-unital channels is $A$-unital, and parallel concatenation of an $A'$-unital channel and an $A$-unital channel is $A'A$-unital \cite{Chitambar2019}. However, this is also readily seen from the definition of $\unip AB$. Let $\M_{AB}, \N_{AB} \in \unip AB$. Then, we have 
\begin{equation}
    \M_{AB} \circ \N_{AB} \left ( \frac{I_A}{d_A} \otimes \rho_B \right)
    = \M_{AB} \left(\frac{I_A}{d_A} \otimes \sigma_B \right)
    = \frac{I_A}{d_A} \otimes \gamma_B. 
\end{equation}
Thus, $\M_{AB} \circ \N_{AB} \in \unip AB$. Similarly, if $\M_{A'B'} \in \unip {A'}{B'}$, and $\N_{AB} \in \unip AB$, from Lemma \ref{lem:auni_aauni}, we have $(\id_{A'B'} \otimes \N_{AB}), (\M_{A'B'} \otimes \id_{AB}) \in \unip {A'A}{B'B}$. Then, from the closedness under serial concatenation above, it follows that
\begin{equation}
    (\id_{A'B'} \otimes \N_{AB}) \circ (\M_{A'B'} \otimes \id_{AB}) = \M_{A'B'} \otimes \N_{AB} \in \unip {A'A}{B'B}.
\end{equation}

\subsection{Convexity}
\label{subsec:convex}
\begin{theorem}
    A convex combination of $A$-unital channels is an $A$-unital channel. 
\end{theorem}
\begin{proof}
    Let $\N_1, \N_2 \in \unip AB$. Then we have, 
    \begin{align}
        \begin{split}
            \left(p \N_1 + (1-p)\N_2 \right) & \left(\frac{I_A}{d_A} \otimes \beta_B \right) \\
                                             &= p \N_1 \left(\frac{I_A}{d_A} \otimes \beta_B \right) + (1-p) \N_2 \left(\frac{I_A}{d_A} \otimes \beta_B \right) \\
                                             &= p \left(\frac{I_A}{d_A} \otimes \gamma_{1B}\right) + (1-p)\left(\frac{I_A}{d_A} \otimes \gamma_{2B}\right) \\
                                             &= \frac{I_A}{d_A} \otimes \left(p \gamma_1 + (1-p) \gamma_2 \right)_B
        \end{split}
    \end{align}
    Thus, $p\N_1 + (1-p)\N_2 \in \unip AB$.
\end{proof}

\subsection{Detection}
\label{subsec:detection}
The characterization of both conditional entropy non-decreasing channels and completely free channels for $\cvenn AB$ as $A$-unital channels makes apparent how these channels can be detected. For any $\St(B)$, it is always possible to find $d_B^2$ density matrices $\{\beta_i\}_i$ from $\St(B)$ that form a basis for for $\St(B)$ \cite{Wilde2013}. For some unknown channel $\N$ given as a blackbox, a necessary and sufficient condition for $A$-unitality is $\N \left( \frac{I_A}{d_A} \otimes \beta_{iB} \right) = \frac{I_A}{d_B} \otimes \beta'_{iB}$ for all the $d_B^2$ states $\beta_i$ from the basis. This follows using the same proof technique as in Lemma \ref{lem:auni_aauni}.   

\section{Relationship between A-Unital Channels and Other Classes}
\label{sec:aunital_sep}

We now turn our attention to the relationship between $A$-unital channels and other classes. If we detect that a channel is $A$-unital (or not $A$-unital), can anything be said about the membership of the channel in $\sepp AB$? It turns out that this is not the case, as there exist $A$-unital channels that are not separable, separable channels that are not $A$-unital, as well as channels in the intersection. We discuss this in the following subsections. We also discuss the relationship between entanglement breaking channels and $A$-unital channels. These relationships are depicted in Fig.~\ref{fig:sep_auni}. Throughout this section, we assume $d_A = d_B = d.$ 

\begin{figure}
    \centering
    \includegraphics[width=0.5\textwidth]{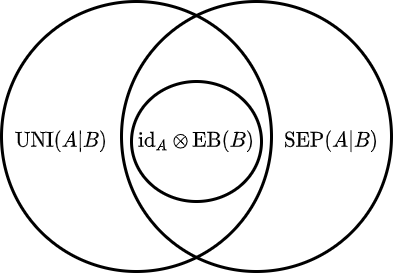}
    \caption{The relationship between $\unip AB$, $\sepp AB$ and $\id_A \otimes \ebp B$ when $d_A = d_B$.}
    \label{fig:sep_auni}
\end{figure}

\subsection{A-Unital Non-Separable Channels}

We prove the existence of channels that belong to $\unip AB$ but not $\sepp AB$ by constructing such a channel. Let $\{\ket i \}_i$ be an orthonormal basis for $\Hil^A$, and $\{\ket j\}_j$ be an orthonormal basis for $\Hil^B$.
\begin{definition}
    \label{def:swap_prepare}
    The Swap and Prepare channel, denoted by $\Sp_{AB}:\St(AB) \rightarrow \St(AB)$ is defined as the composition of the following channels:
    \begin{itemize}
        \item The Swap Channel $[Sw_{AB}: \St(AB) \rightarrow \St(AB)]$. The unitary channel that swaps systems $A$ and $B$. Its Kraus decomposition contains just one Kraus operator which is the unitary matrix $U_s \equiv \sum_{ij} \ket {ji} \bra{ij}$. We have \begin{equation} Sw_{AB}(\rho) = U_s \rho U_s^\dagger. \end{equation} 

        \item The Partial Trace Channel $[\Tr_A: \St(AB) \rightarrow \St(B)]$. The Kraus operators for this channel are $\left\{\bra i_A \otimes I_B \right\}_i$. We have  \begin{equation}\Tr_A(\rho) = \sum_i \left(\bra i_A \otimes I_B \right)(\rho_{AB}) \left(\ket i_A \otimes I_B \right). \end{equation}

        \item The Preparation of $\frac{I_A}{d}$ Channel $[Pr_A: \St(B) \rightarrow \St(AB)]$. The Kraus operators for this channel are $\left\{\frac{1}{\sqrt {d}} \ket i_A \otimes I_B\right\}_i$. We have \begin{align} Pr_A(\rho_B) &= \sum_i \frac{1}{d} \left(\ket i_A \otimes I_B \right) (\rho_B) \left(\bra i_A \otimes I_B \right) \\ 
            &= \frac{I_A}{d} \otimes \rho_B \end{align}
    \end{itemize}

    Finally, we define the Swap and Prepare Channel:
    \begin{equation}
        \Sp_{AB}(\rho_{AB}) = \left(Pr_A \circ \Tr_A \circ Sw_{AB}\right)(\rho_{AB}) 
    \end{equation}
\end{definition}

\begin{theorem}
    \label{thm:aunital_nonsep}
    There exist A-unital non-separable channels.
\end{theorem}

\begin{proof}
    In Lemma \ref{lem:ex_auni}, we show that the channel $\Sp_{AB} \in \unip AB$, and in Lemma \ref{lem:ex_nonsep}, we show that $\Sp_{AB} \notin \sepp AB$.  
\end{proof}

\begin{lemma}
    \label{lem:ex_auni}
    $\Sp_{AB} \in \unip AB$.
\end{lemma}
\begin{proof}
    Let $\sigma_B \in \B(B)$. We have 
    \begin{align}
        \begin{split}
            \Sp_{AB} \left(\frac{I_A}{d} \otimes \sigma_B\right)  
                                       &= Pr_A \circ \Tr_A \circ Sw_{AB} \left(\frac{I_A}{d} \otimes \sigma_B\right) 
                                       = Pr_A \circ \Tr_A \left(\sigma_A \otimes \frac{I_B}{d}\right) \\
                                       &= Pr_A \left( \frac{I_B}{d}\right) 
                                       = \frac{I_A}{d} \otimes \frac{I_B}{d}.
        \end{split}
    \end{align}
    Thus, using Def.~\ref{def:auni}, $\Sp_{AB} \in \unip AB$.
\end{proof}

\begin{lemma}
    \label{lem:ex_nonsep}
    $\Sp_{AB} \notin \sepp AB$.
\end{lemma}
\begin{proof}
    From the discussion following Def.~\ref{def:sep}, we have that $\sepp AB = \sepscompfree AB$. Thus, it is sufficient to prove that $\Sp_{AB} \notin \sepscompfree AB$. Consider Hilbert spaces $\Hil^{A'}$ and $\Hil^{B'}$ where $d_{A'}$ and $d_{B'}$ are also equal to $d$. Take the state $\chi \in \St(A'AB'B)$ defined as $\chi \equiv \ketbra {\phi^+}{\phi^+}_{A'A} \otimes \ketbra {00}{00}_{B'B}$, where $\ketbra {\phi^+}{\phi^+}_{A'A}$ is the entangled state $\frac{1}{d} \sum_{ij} \ketbra {i}{j}_{A'} \otimes \ketbra {i}{j}_{A}$. The state $\chi$ is a product state across the $A'A|B'B$ cut, and hence $\chi \in \seps {A'A}{B'B}.$ We now apply $\id_{A'B'} \otimes \Sp_{AB}$ to the state $\chi$.
    \begin{align}
    \begin{split}
        \chi' &= \left( \id_{A'B'} \otimes \Sp_{AB} \right) (\chi) \\
              &= \left(\id_{A'B'} \otimes \Sp_{AB} \right) \left(\frac{1}{d}\sum_{ij} \ketbra {i}{j}_{A'} \otimes \ketbra {i}{j}_A \otimes \ketbra {0}{0}_{B'} \otimes \ketbra {0}{0}_B \right)  \\
              &=\left(\id_{A'B'} \otimes \: Pr_A \circ \Tr_A \circ Sw_{AB}\right)\left(\frac{1}{d}\sum_{ij} \ketbra {i}{j}_{A'} \otimes \ketbra {i}{j}_A \otimes \ketbra {0}{0}_{B'} \otimes \ketbra {0}{0}_B \right)  \\
              &=\left(\id_{A'B'} \left(\frac{1}{d}\sum_{ij} \ketbra {i}{j}_{A'} \otimes \ketbra {0}{0}_{B'}\right)\right) \otimes \left( \left(Pr_A \circ \Tr_A \circ Sw_{AB}\right)\left( \ketbra {i}{j}_A \otimes \ketbra {0}{0}_B \right)\right) \\
               &=\left(\frac{1}{d}\sum_{ij} \ketbra {i}{j}_{A'} \otimes \ketbra {0}{0}_{B'}\right) \otimes \left( (Pr_A \circ \Tr_A)\left( \ketbra {0}{0}_A \otimes \ketbra {i}{j}_B \right)\right)  \\
               &=\left(\frac{1}{d}\sum_{ij} \ketbra {i}{j}_{A'} \otimes \ketbra {0}{0}_{B'}\right) \otimes \left( Pr_A\left( \ketbra {i}{j}_B \right)\right)  \\
              &= \frac{1}{d}\sum_{ij} \ketbra {i}{j}_{A'} \otimes \ketbra {0}{0}_{B'} \otimes \frac{I_A}{d} \otimes \ketbra {i}{j}_B \\
        &= \ketbra{\phi^+}{\phi^+}_{A'B} \otimes \frac{I_A}{d}\otimes \ketbra{0}{0}_{B'}.
        \end{split}
    \end{align}
    Observe that $\chi'$ is entangled across the $A'A|B'B$ cut. Thus, $\chi' \notin \seps {A'A}{B'B}$. From the definition of $\sepscompfree AB$, this implies that $\Sp_{AB} \notin \sepscompfree AB$, hence $\Sp_{AB} \notin \sepp AB$. The action of $\Sp_{AB}$ is depicted in Fig.~\ref{fig:prepare_and_swap}.
\end{proof}
From Lemmas \ref{lem:ex_auni} and \ref{lem:ex_nonsep}, we have that $\Sp_{AB} \in \unip AB$ but $\Sp_{AB} \notin \sepp AB$, and thus, Theorem \ref{thm:aunital_nonsep} follows.

\begin{figure}
    \centering
    \includegraphics[width=0.7\linewidth]{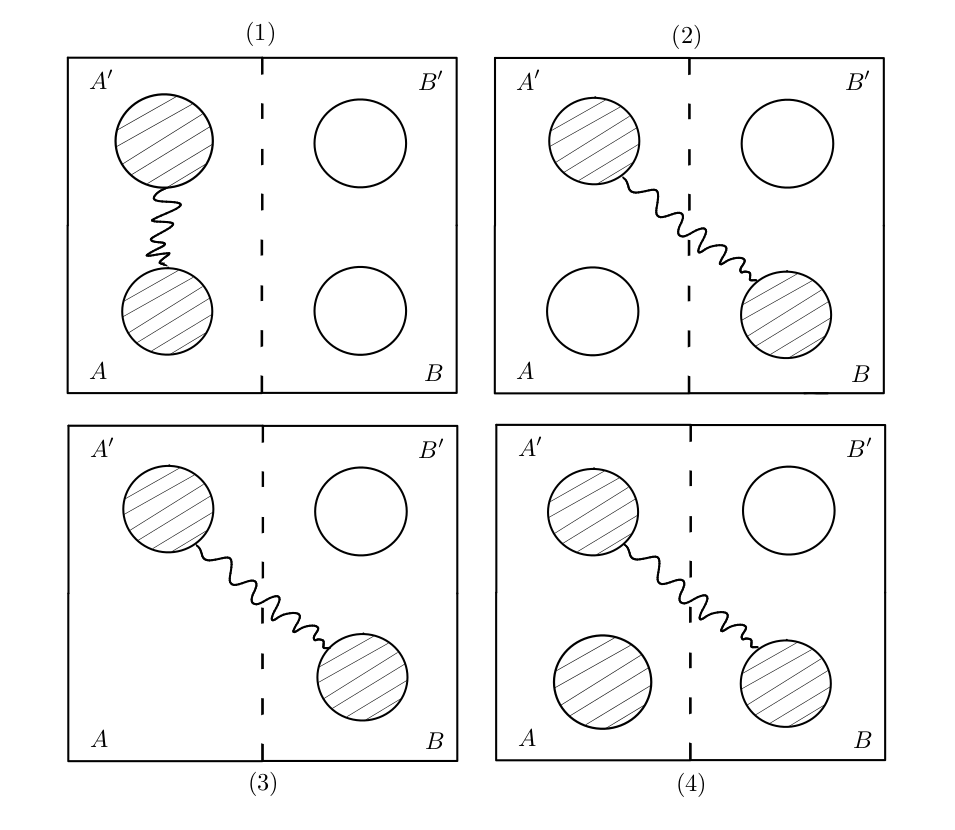}
    \caption{Demonstration of $\Sp_{AB}$ acting on $\chi$. Subfigure $(1)$ depicts the state $\chi$, which is separable across the $A'A|B'B$ cut. Upon applying $\id_{A'B'} \otimes Sw_{AB}$, the states of subsystems $A$ and $B$ are swapped, as shown in subfigure $(2)$. Application of $\Tr_A$ is depicted in subfigure $(3)$, and finally, $Pr_A$ prepares the state $\frac{I_A}{d}$ and this is depicted in subfigure $(4)$. Thus, $\Sp_{AB}$ is $A$-unital, but not separable, as the channel generates entanglement across the $A'A|B'B$ cut.}
    \label{fig:prepare_and_swap}
\end{figure}
\subsection{Non-A-Unital Separable Channels}

Examples of channels that are non-A-unital but separable are replacement channels \cite{Chitambar2019} $R$  which trace out a given state and prepare a product state $\rho_A \otimes \rho_B$ where $\rho_A \neq \frac{I_A}{d}$. Since for all $\sigma_B \in \St(B)$, we have $R \left(\frac{I_A}{d} \otimes \sigma_B \right ) \neq \frac{I_A}{d} \otimes \gamma_B$ for some $\gamma \in \St(B)$, $R$ is not an $A$-unital channel. However, since both tracing out as well as preparation of a product state are separable channels, $R$ is a separable channel. A specific example of a channel which prepares the state $\ketbra {00}{00}_{AB}$ is
\begin{equation}
    R_0(\rho_{AB}) = \sum_{ij} \left( \ketbra {0}{i}_A \otimes \ketbra {0}{j}_B \right) (\rho_{AB}) \left (\ketbra {i}{0}_A \otimes \ketbra {j}{0}_B \right),
\end{equation}
where $\{\ket i\}_i$ and $\{\ket j\}_j$ are orthonormal bases for $\Hil^A$ and $\Hil^B$ respectively. 
\subsection{A-Unital Separable Channels}

A trivial example of a channel that is both $A$-unital and separable is the identity channel. A more non-trivial example is any channel of the form  
\begin{equation}
    C_{AB}(\rho) = \sum_i \left (K^i_A \otimes L^i_B \right ) \rho \left ( K^i_A \otimes L^i_B \right )^\dagger
\end{equation}
where $\sum_i \left ( K^i_A \otimes L^i_B \right )^\dagger \left (K^i_A \otimes L^i_B \right ) = I_{AB}$ and each $K^i_A$ is a unitary channel. $C_{AB}$ is separable by definition. It is also $A$-unital as shown below. Let $\sigma \in \St(B)$. We have
\begin{align}
    \begin{split}
        C_{AB}\left(\frac{I_A}{d} \otimes \sigma_B\right) &= \sum_i \left (K^i_A \otimes L^i_B \right ) \left( \frac{I_A}{d} \otimes \sigma_B \right) \left ( K^i_A \otimes L^i_B \right )^\dagger \\
                                                          &= \sum_i \frac{1}{d}K^i_A I_A K^{i \dagger}_A \otimes L^i_B \sigma_B L^{i \dagger}_B = \frac{I_A}{d} \otimes \sum_i L^i_B \sigma_B L^{i \dagger}_B.
    \end{split}
\end{align}
\subsection{A-Unital and Entanglement-Breaking Channels}

Both $\unip AB$ and $\sepp AB$ are concerned with correlations across the $A|B$ split. Thus, we relate both these classes with channels of the form $\id_A \otimes \Phi_B$ where $\Phi_B \in \ebp B$, that is $\Phi_B$ is an entanglement-breaking channel (Def.~\ref{def:eb}), as they destroy all entanglement across the $A|B$ split. These channels are $A$-unital, as $(\id_A \otimes \Phi_B) \left( \frac{I_A}{d}\otimes \sigma_B \right) = \frac{I_A}{d} \otimes \Phi(\sigma)_B$ for all $\sigma_B \in \St(B)$, as well as separable, as the Kraus operators for the channel are $\{I_A \otimes K^i_B\}_i$ where $\{K^i_B\}_i$ are the Kraus operators for $\Phi_B$. This is depicted in Figure \ref{fig:sep_auni}. On the contrary, channels of the form $\Phi_A \otimes \id_B$ where $\Phi_A \in \ebp A$ are not necessarily $A$-unital, although they are separable. As an example, consider the following classical-quantum channel $\M_A : \St(A) \rightarrow \St(A)$, that acts as  
\begin{equation}
    \M_A (\rho) = \sum_i \bra i \rho \ket i \ketbra {0}{0}_A.
\end{equation}
\sloppy All classical-quantum channels are entanglement-breaking \cite{Wilde2013}, however, the channel $\M_A \otimes \id_B$ is not $A$-unital, as we have $\left (\M_A \otimes \id_B \right) \left(\frac{I_A}{d} \otimes \sigma_B \right) = \ketbra{0}{0} \otimes \sigma_B$ which is not of the form $\frac{I_A}{d} \otimes \gamma_B$ for some $\gamma_B \in \St(B)$. 

$A$-unital channels are closely related to unital channels. Indeed, as given in Def.~\ref{def:unital}, while $A$-unital channels map states of the form $\frac{I_A}{d_A} \otimes \sigma_B$ to $\frac{I_A}{d_A} \otimes \sigma'_B$, unital channels map $\frac{I_{AB}}{d_{AB}}$ to $\frac{I_{AB}}{d_{AB}}$. In the next section, we explore this connection further. We show that while $A$-unital channels are the largest class of completely free channels for $\cvenn AB$, unital channels form the largest class of completely free channels for a previously defined, related set $\acvenn AB$, the set of all states whose conditional entropy remains non-negative under global unitary operations.  The relationship between $A$-unital, $B$-unital and unital channels is depicted in Fig.~\ref{fig:auni_buni_uni}.

\section{Unital Operations and Absolute Non-Negative Conditional Entropy}
\label{sec:unital}
\subsection{Absolutely Non-Negative Conditional Entropy States}
Previously, the absolute version of $\cvenn AB$ (analogous to absolutely separable states \cite{Kuifmmode2001, Halder2021}) was defined and characterized for the case where $d_A = d_B = 2$ \cite{Patro2017}. Here, we extend the definition by characterizing it for the case where $d_A = d_B = d$. 

\begin{definition}
    \label{def:acvenn} 
    A quantum state $\rho \in \St(AB)$, where $d_A = d_B = d$, which satisfies $\cve AB(U\rho U^\dagger) \geq 0$ for all unitary channels $U: \St(AB) \rightarrow \St(AB)$ is called an absolutely non-negative conditional entropy state. The set of absolutely non-negative conditional entropy states in $\St(AB)$ is denoted by $\acvenn AB$.
\end{definition}

The set $\acvenn AB$ is well-characterized in terms of the von Neumann entropy, as shown in the following theorem.

\begin{theorem}
    \label{thm:acvenn_def}
    $\rho_{AB} \in \acvenn AB$ if and only if $\ent (\rho_{AB}) \geq \log d$.
\end{theorem}

\begin{proof}
    Let $\rho_{AB} \in \St(AB)$ and $\ent(\rho_{AB}) \geq \log d$. As the maximum von Neumann entropy of the subsystem $B$ is upperbounded by $\log d$, we have $\forall \, U, \, \ent(U\rho_{AB} U^\dagger) = \ent(\rho_{AB}) \geq \log d \geq \ent \left((U\rho_{AB} U^\dagger)_B\right)$. Therefore, $\forall \, U, \, \cve {A}{B}(U\rho U^\dagger) = \ent(U\rho U^\dagger) - \ent\left((U \rho_{AB} \rho^\dagger)_B\right) \geq 0$ and $\rho_{AB} \in \acvenn AB$.   

        To prove the ``only if" part of the statement, consider the following set of vectors $\{\ket {\psi_{nm}}\}_{nm}$ which is a basis for the vector space $\Hil^A \otimes \Hil^B$ \cite{Bennett1993}
        \begin{equation}
            \label{eq:teleportation_basis}
            \ket {\psi_{nm}} = \sum_j e^{2 \pi i jn /d} \ket j \otimes \ket{(j+m) \mod d} / \sqrt d.
        \end{equation}
        Notice that for each $\ket {\psi_{nm}}$, we have $\left(\ket {\psi_{nm}} \bra {\psi_{nm}}\right)_B = \frac{I}{d}$. Now, let $\sigma_{AB} \in \acvenn AB$. Let the set of vectors $\{\ket{\gamma_{nm}}\}_{nm}$ be the eigenbasis of $\sigma_{AB}$. Consider the unitary channel $U_e \equiv \sum_{nm} \ket{\psi_{nm}} \bra{\gamma_{nm}}$. The state $U_e \sigma_{AB} U_e^\dagger$ is diagonal in the $\ket{\psi_{nm}}$ basis, hence, $\left(U_e \sigma_{AB} U_e^\dagger\right)_B = \frac{I}{d}$, and $\ent\left((U_e \sigma_{AB} U_e^\dagger)_B\right) = \log d$. 

        By definition, we have $\cve {A}{B}(U \sigma_{AB} U^\dagger) \geq 0$ for all $U$, in particular for $U_e$. Therefore, $\cve {A}{B}(U_e \sigma_{AB} U_e^\dagger) = \ent(U_e \sigma_{AB} U_e^\dagger) - \ent\left((U_e \sigma_{AB} U_e^\dagger)_B\right)  \geq 0$, which implies that $\ent(\sigma_{AB}) - \log d \geq 0$. Thus, $\, \forall \, \sigma_{AB} \in \acvenn AB$, we have $\ent(\sigma_{AB}) \geq \log d$.   
\end{proof}

\begin{figure}
\begin{center}
    \includegraphics[width=0.5\textwidth]{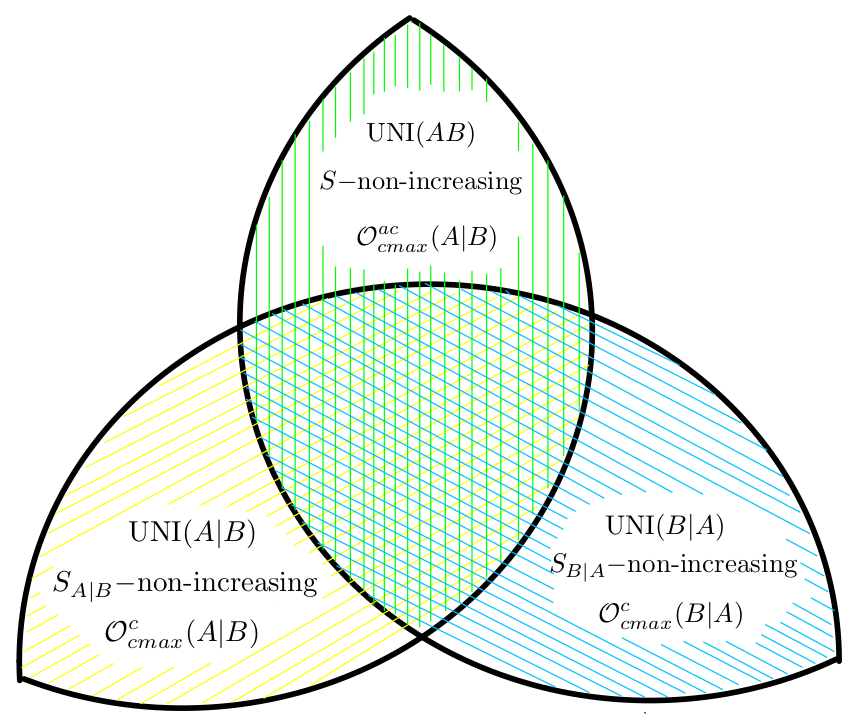}
    \caption{The relationship between $A$-unital, $B$-unital and unital channels when $d_A = d_B$.}
    \label{fig:auni_buni_uni}
\end{center}
\end{figure}

\subsection{Completely Free Operations for Absolute Non-Negative Conditional Entropy States}
Here, we define the class of completely free operations for $\acvenn AB$, respecting the $d_A = d_B = d$ condition in the definition of $\acvenn AB$. 

\begin{definition}
    \label{def:acvenncompfree}
    A channel $\N_{AB}: \St(AB) \rightarrow \St(AB)$ is completely free for $\acvenn AB$ if for all finite dimensional Hilbert spaces $\Hil^{A'}$  and $\Hil^{B'}$, where $d_{A'} = d_{B'} = d'$ and $\rho \in \acvenn {A'A}{B'B}$, we have  $\left( \id_{A'B'} \otimes \N_{AB} \right)(\rho) \in \acvenn {A'A}{B'B}$. The class of all such completely free operations mapping $\St(AB)$ to itself is denoted by $\acvenncompfree AB$. 
\end{definition}

We now prove that for the above definition of $\acvenncompfree AB$, unitality is a necessary and sufficient condition for complete freeness. 

\begin{theorem}
    \label{thm:unital_acvenn}
    A channel $\N \in \acvenncompfree AB$ if and only if it is unital.
\end{theorem}

\begin{proof} 
    Let $\Hil^A, \Hil^B, \Hil^{A'}$ and $\Hil^{B'}$ be finite dimensional Hilbert spaces such that $d_A = d_B = d$ and $d_{A'} = d_{B'} = d'$. Let $\N_{AB}$ be a unital channel mapping $\St(AB)$ to $\St(AB)$. Then, notice that the channel $\N' = \id_{A'B'} \otimes \N_{AB}$ is also a unital channel mapping $\St(A'AB'B)$ to itself. This unitality implies we have $\ent\left(\N'(\rho)\right) \geq \ent\left(\rho \right) \, \forall \, \rho \in \St(A'AB'B)$ \cite{Wilde2013}. Hence, $\forall \, \sigma \in \acvenn {A'A}{B'B}$, we have $\ent\left(\N'(\sigma)\right) \geq \ent\left(\sigma \right) \geq \log dd'$, using Theorem \ref{thm:acvenn_def}. Thus, $\N'(\sigma) \in \acvenn {A'A}{B'B}$. This implies that $\N'$ is a free operation, which in turn implies that $\N_{AB} \in \acvenncompfree AB$.

    Let $\M_{AB}$ be a non-unital channel in $\St(AB)$. Hence, $\M_{AB}\left(\frac{I_{AB}}{d^2}\right) = \sigma_{AB}$ where $\sigma_{AB} \neq \frac{I_{AB}}{d^2}$. Thus, $\ent(\sigma_{AB}) < \log d^2$, since the maximum entropy for $\rho \in \St(AB) = \log d^2$ and is achieved only at $\frac{I_{AB}}{d^2}$. Consider the state $\chi = \gamma_{A'B'} \otimes \frac{I_{AB}}{d^2}$ where $\gamma$ is a pure state and $d' = d$. Hence $\ent(\chi) = \log d^2$ which implies that $\chi \in \acvenn {A'A}{B'B}$. We have $\chi' = \left(\id_{A'B'} \otimes \M_{AB}\right)\left(\chi\right)  = \gamma_{A'B'} \otimes \sigma_{AB}$. Therefore, $\ent(\chi') < \log d^2$ and $\chi' \notin \acvenn {A'A}{B'B}$. This implies that $\id_{A'B'} \otimes \M_{AB}$ is not a free operation, and therefore, $\M_{AB} \notin \acvenncompfree AB.$ 

\end{proof}

\begin{figure}
    \centering
    \includegraphics[width=0.8\textwidth]{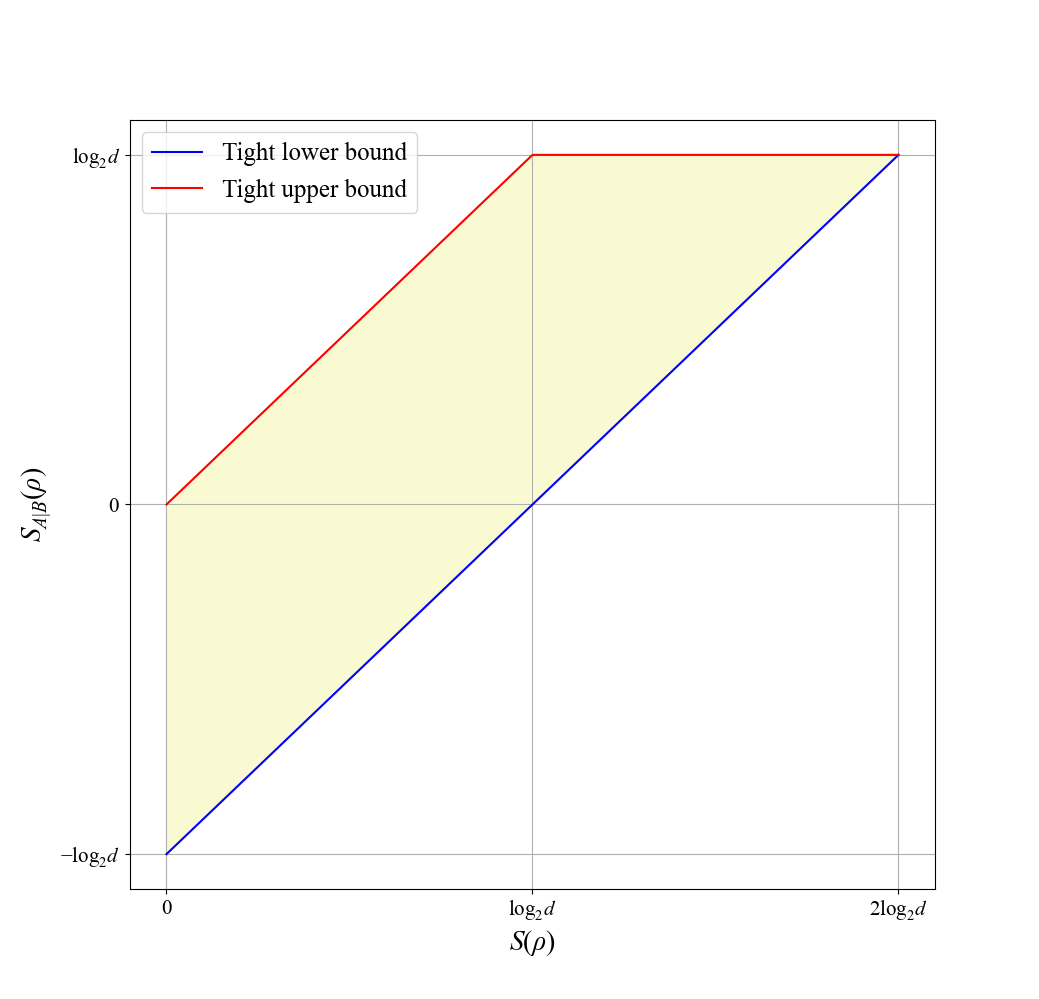}
    \caption{Upper and lower bounds for $\cve AB(\rho)$ given $\ent(\rho)$.}

    \label{fig:upper_lower_bounds}
\end{figure}

\section{Upper and Lower Bounds for Conditional Entropy given Entropy}
\label{sec:bounds}

The characterization of $\acvenn AB$ via entropy suggests that the entropy of a state can determine the values of conditional entropy that the state can possess. In the following section, we discuss upper and lower bounds for conditional entropy, given the entropy of a state $\rho_{AB}$, where $d_A = d_B = d$.

\begin{theorem}
    \label{thm:bounds}
    Let $E_x \equiv \{\cve AB(\rho) \; | \; \ent(\rho) = x \; \forall \; \rho \in \St(AB)\}$. Then, the tight lower bound of $E_x$ is $x - \log d$ and the tight upper bound of $E_x$ is $\max(x, \log d).$
\end{theorem}

\begin{proof}
    Fig.~\ref{fig:upper_lower_bounds} depicts the bounds of the above theorem. The lower bound is proven in Lemma \ref{lem:lower_bound} and the upper bound is proven in Lemmas \ref{lem:first_upper} and \ref{lem:second_upper}.
\end{proof}

\begin{lemma}
    \label{lem:lower_bound}
    $\forall x \in [0, 2\log d]$, the tight lower bound of $E_x \equiv \{\cve AB(\rho) \mid \ent(\rho) = x \land \rho \in \St(AB) \}$ is $x - \log d$.
\end{lemma}

\begin{proof}
    Consider $\rho \in \St(AB) \mid \ent(\rho) = x$. Since $\ent(\rho_B) \leq \log d$ (Theorem 11.8 in \cite{Nielsen2000}), we have $\cve AB (\rho) \geq x - \log d$. Thus, 
    \begin{equation} 
        \label{lb_purity}
        \text{a lower bound of }E_x \text{ is } x - \log d. 
    \end{equation}

    It now remains to show that $\forall x \in [0, 2\log d]$, $x - \log d \in E_x$. Consider the function $\R: [0,1] \rightarrow \St(AB) : p \rightarrow p |\phi^+\rangle \langle \phi^+| + (1-p) \frac{I}{d^2}$ where $|\phi^+\rangle = \sum_{i=0}^{d-1}\frac{|ii\rangle}{\sqrt d}$. Note that $\forall p \in [0,1]$, we have $\R(p)_B = \frac{I}{d}$. Thus,
    \begin{equation}
        \label{s_b}
        \ent(\R(p)_B) = \log d.
    \end{equation}
    As $\ent : \St(AB) \rightarrow [0, 2\log d]$ is a continuous function (Box 11.2 in \cite{Nielsen2000}) and $\R$ can be shown to be a continuous function \footnote{This can be shown using the $\epsilon-\delta$ definition of continuity.}, we have $\ent \circ \R: [0,1] \rightarrow [0, 2\log d]$ is a continuous function. We have $\ent(\R(0)) = 2\log d$ and $\ent(\R(1)) = 0$. Thus, by the intermediate value theorem, 
    \begin{equation}
        \label{inter}
        \forall x \in [0, 2\log d],  \; \exists \; p_x \in [0,1] \mid \ent(\R(p_x)) = x.
    \end{equation}

    Using Eq. \ref{s_b} and \ref{inter}, we have
    \begin{equation}
        \label{attain}
        \cve AB(\R(p_x)) = \ent(\R(p_x)) - \ent(\R(p_x)_B) = x - \log d.
    \end{equation}
    $\ent(\R(p_x)) = x$ (Eq.~\ref{inter}) implies that $\cve AB(\R(p_x)) = x - \log d \in E_x$ (Eq.~\ref{attain}). Thus, as $x - \log d$ is a lower bound of $E_x$, and $x - \log d \in E_x$, it is the tight lower bound.
\end{proof}

\begin{lemma}
    \label{lem:first_upper}
    $\forall x \in [0, \log d]$, the tight upper bound of $E_x \equiv \{\cve AB(\rho) \mid \ent(\rho) = x \land \rho \in \St(AB) \}$ is $x$.
\end{lemma}

\begin{proof}
    Consider $\rho \in \St(AB) \mid \ent(\rho) = x$. Since $\ent(\rho_B) \geq 0$, we have $\cve AB(\rho) \leq x$. Thus, an upper bound of $E_x$ is $x$. It now remains to show that $\forall x \in [0, \log d]$, $x \in E_x$. For a given $x \in [0, \log d]$, consider states of the form $\sigma = \sigma_A \otimes \sigma_B$, where $\ent(\sigma_A) = x$ and $\sigma_B$ is any pure state. We have $\ent(\sigma) = x+0 = x$, which implies $\cve AB(\sigma) = x - 0 = x \in E_x$. 

Therefore, since $x$ is an upperbound of $E_x$ and $x \in E_x$, it is a tight upper bound.
\end{proof}

\begin{lemma}
    \label{lem:second_upper}
    $\forall x \in [\log d, 2\log d]$, the tight upper bound of $E_x \equiv \{\cve AB(\rho) \mid \ent(\rho) = x \land \rho \in \St(AB)\}$ is $\log d$.
\end{lemma}
\begin{proof}
    From Theorem 11.4.1 in \cite{Wilde2013} and Theorem 11.8 in \cite{Nielsen2000}, we have $\forall \rho \in \St(AB), \; \cve AB(\rho) \leq \ent(\rho_A) \leq \log d$. Therefore, $\log d$ is also an upper bound for $E_x$. It now remains to show that $\forall x \in [\log d, 2 \log d], \; \log d \in E_x$. Consider states of the form $\sigma = \frac{I}{d} \otimes \sigma_B$ where $\ent(\sigma_B) = x - \log d$. We have $\ent(\sigma) = x$ which implies that $\cve AB(\sigma) = \log d \in E_x$. 
    Therefore, since $\log d$ is an upper bound for $E_x$ and $\log d \in E_x$, it is a tight upper bound.
\end{proof}

Theorem \ref{thm:bounds} follows from the lower bound proved in Lemma \ref{lem:lower_bound} and the upper bounds proved in  Lemmas \ref{lem:first_upper} and \ref{lem:second_upper}.

\section{Conclusion}
\label{sec:conclusion}
In this paper, we introduce the class of $A$-unital channels, and show that this is equal to the largest class of completely free channels for the set of states with non-negative conditional entropy $\left(\cvenn AB \right)$. We also show that this is the largest class of channels that do not decrease the conditional entropy of a state. We discuss the relation between $A$-unital channels, separable channels and entanglement-breaking channels. Furthermore, we extended the characterization of the set of states that possess non-negative conditional entropy even under global unitary transformations ($\acvenn AB$) from the $2 \otimes 2$ case to the $d \otimes d$ case. We demonstrate that unital channels are the largest class of completely free operations for this set. Finally, we provide upper and lower bounds for the conditional entropy of a state as a function of the entropy of a state. We conclude with some questions for further research that emerge from our work. 

\begin{enumerate}

    \item \textbf{Characterization of $\acvenn AB$ for the $d_A \otimes d_B$ case.} In this work, we show that when $d_A = d_B = d$, $\acvenn AB$ is the set of states in $\St(AB)$ for which $\ent(\rho_{AB}) \geq \log d$. An interesting problem is the extension of this characterization to the case $d_A \neq d_B$. This is important for the full characterization of completely free operations for $\acvenn AB$.    

    \item \textbf{Decomposition of A-unital channels and a Choi matrix representation.} A necessary and sufficient condition for the unitality of a channel is for it to be expressed as an affine combination of unitary quantum channels \cite{Mendl2009}. Providing a similar decomposition for A-unital channels is an open problem. Moreover, separable channels, completely PPT-preserving channels and entanglement-breaking channels all have elegant characterizations in terms of Choi matrices \cite{Cirac2001,Rains1999,Rains2001, Regula2019, Wilde2013, Jiang2013}. A similar characterization remains to be given for $A$-unital channels.  
     
    \item \textbf{State conversion under A-unital channels.} According to Uhlmann's theorem, a state $\rho$ can be converted to state $\sigma$ under unital channels if and only if $\sigma \prec \rho$, that is $\rho$ majorizes $\sigma$ \cite{Nielsen2002, Li2013}. What are the conditions for state conversion under $A$-unital channels?    
    
\end{enumerate}

\paragraph{\textit{Note added:}} \textit{Also see related independent work by Brandsen et al. \cite{Brandsen2021}}

\acknowledgements{We gratefully acknowledge Prof. Arun K. Pati for his enriching discussions, and Mohammad Alhejji for his fruitful comments.}

\onecolumn\newpage
\appendix


\begin{thebibliography}{39}
\providecommand{\natexlab}[1]{#1}
\providecommand{\url}[1]{\texttt{#1}}
\expandafter\ifx\csname urlstyle\endcsname\relax
  \providecommand{\doi}[1]{doi: #1}\else
  \providecommand{\doi}{doi: \begingroup \urlstyle{rm}\Url}\fi

\bibitem[Horodecki et~al.(2009)Horodecki, Horodecki, Horodecki, and
  Horodecki]{Horodecki2009}
Ryszard Horodecki, Pawe\l{} Horodecki, Micha\l{} Horodecki, and Karol
  Horodecki.
\newblock Quantum entanglement.
\newblock \emph{Rev. Mod. Phys.}, 81:\penalty0 865--942, Jun 2009.
\newblock \doi{10.1103/RevModPhys.81.865}.
\newblock URL \url{https://link.aps.org/doi/10.1103/RevModPhys.81.865}.

\bibitem[Chitambar and Gour(2019)]{Chitambar2019}
Eric Chitambar and Gilad Gour.
\newblock Quantum resource theories.
\newblock \emph{Rev. Mod. Phys.}, 91:\penalty0 025001, Apr 2019.
\newblock \doi{10.1103/RevModPhys.91.025001}.
\newblock URL \url{https://link.aps.org/doi/10.1103/RevModPhys.91.025001}.

\bibitem[Chitambar et~al.(2014)Chitambar, Leung, Man{\v{c}}inska, Ozols, and
  Winter]{Chitambar2014}
Eric Chitambar, Debbie Leung, Laura Man{\v{c}}inska, Maris Ozols, and Andreas
  Winter.
\newblock Everything you always wanted to know about {LOCC} (but were afraid to
  ask).
\newblock \emph{Communications in Mathematical Physics}, 328\penalty0
  (1):\penalty0 303--326, March 2014.
\newblock \doi{10.1007/s00220-014-1953-9}.
\newblock URL \url{https://doi.org/10.1007/s00220-014-1953-9}.

\bibitem[Cirac et~al.(2001)Cirac, D\"ur, Kraus, and Lewenstein]{Cirac2001}
J.~I. Cirac, W.~D\"ur, B.~Kraus, and M.~Lewenstein.
\newblock Entangling operations and their implementation using a small amount
  of entanglement.
\newblock \emph{Phys. Rev. Lett.}, 86:\penalty0 544--547, Jan 2001.
\newblock \doi{10.1103/PhysRevLett.86.544}.
\newblock URL \url{https://link.aps.org/doi/10.1103/PhysRevLett.86.544}.

\bibitem[Vedral and Plenio(1998)]{Vedral1998}
V.~Vedral and M.~B. Plenio.
\newblock Entanglement measures and purification procedures.
\newblock \emph{Phys. Rev. A}, 57:\penalty0 1619--1633, Mar 1998.
\newblock \doi{10.1103/PhysRevA.57.1619}.
\newblock URL \url{https://link.aps.org/doi/10.1103/PhysRevA.57.1619}.

\bibitem[Horodecki et~al.(2003)Horodecki, Shor, and Ruskai]{Horodecki2003}
Michael Horodecki, Peter~W. Shor, and Mary~Beth Ruskai.
\newblock Entanglement breaking channels.
\newblock \emph{Reviews in Mathematical Physics}, 15\penalty0 (06):\penalty0
  629--641, August 2003.
\newblock \doi{10.1142/s0129055x03001709}.
\newblock URL \url{https://doi.org/10.1142/s0129055x03001709}.

\bibitem[Macchiavello and Rossi(2013)]{Macchiavello2013}
C.~Macchiavello and M.~Rossi.
\newblock Quantum channel detection.
\newblock \emph{Phys. Rev. A}, 88:\penalty0 042335, Oct 2013.
\newblock \doi{10.1103/PhysRevA.88.042335}.
\newblock URL \url{https://link.aps.org/doi/10.1103/PhysRevA.88.042335}.

\bibitem[Lee and Watrous(2020)]{Lee2020}
Colin Do-Yan Lee and John Watrous.
\newblock Detecting mixed-unitary quantum channels is {NP}-hard.
\newblock \emph{Quantum}, 4:\penalty0 253, April 2020.
\newblock \doi{10.22331/q-2020-04-16-253}.
\newblock URL \url{https://doi.org/10.22331/q-2020-04-16-253}.

\bibitem[Montanaro and de~Wolf(2016)]{Montanaro2016}
Ashley Montanaro and Ronald de~Wolf.
\newblock A survey of quantum property testing.
\newblock \emph{Theory of Computing}, 1\penalty0 (1):\penalty0 1--81, 2016.
\newblock \doi{10.4086/toc.gs.2016.007}.
\newblock URL \url{https://doi.org/10.4086/toc.gs.2016.007}.

\bibitem[Milazzo et~al.(2020)Milazzo, Braun, and Giraud]{Milazzo2020}
N.~Milazzo, D.~Braun, and O.~Giraud.
\newblock Truncated moment sequences and a solution to the channel separability
  problem.
\newblock \emph{Phys. Rev. A}, 102:\penalty0 052406, Nov 2020.
\newblock \doi{10.1103/PhysRevA.102.052406}.
\newblock URL \url{https://link.aps.org/doi/10.1103/PhysRevA.102.052406}.

\bibitem[Gharibian(2010)]{Gharibian2010}
Sevag Gharibian.
\newblock Strong np-hardness of the quantum separability problem.
\newblock \emph{Quantum Info. Comput.}, 10\penalty0 (3):\penalty0 343–360,
  March 2010.
\newblock ISSN 1533-7146.
\newblock \doi{10.26421/QIC10.3-4-11}.

\bibitem[Cerf and Adami(1997)]{Cerf1997}
N.~J. Cerf and C.~Adami.
\newblock Negative entropy and information in quantum mechanics.
\newblock \emph{Phys. Rev. Lett.}, 79:\penalty0 5194--5197, Dec 1997.
\newblock \doi{10.1103/PhysRevLett.79.5194}.
\newblock URL \url{https://link.aps.org/doi/10.1103/PhysRevLett.79.5194}.

\bibitem[Wilde(2016)]{Wilde2013}
Mark~M. Wilde.
\newblock \emph{Quantum Information Theory}.
\newblock Cambridge University Press, 2016.
\newblock \doi{10.1017/9781316809976}.
\newblock URL \url{https://doi.org/10.1017/9781316809976}.

\bibitem[Ollivier and Zurek(2001)]{Ollivier2001}
Harold Ollivier and Wojciech~H. Zurek.
\newblock Quantum discord: A measure of the quantumness of correlations.
\newblock \emph{Phys. Rev. Lett.}, 88:\penalty0 017901, Dec 2001.
\newblock \doi{10.1103/PhysRevLett.88.017901}.
\newblock URL \url{https://link.aps.org/doi/10.1103/PhysRevLett.88.017901}.

\bibitem[Henderson and Vedral(2001)]{Henderson2001}
L~Henderson and V~Vedral.
\newblock Classical, quantum and total correlations.
\newblock \emph{Journal of Physics A: Mathematical and General}, 34\penalty0
  (35):\penalty0 6899--6905, Aug 2001.
\newblock \doi{10.1088/0305-4470/34/35/315}.
\newblock URL \url{https://doi.org/10.1088%2F0305-4470%2F34%2F35%2F315}.

\bibitem[Radhakrishnan et~al.(2020)Radhakrishnan, Lauri\`ere, and
  Byrnes]{Radhakrishnan2020}
Chandrashekar Radhakrishnan, Mathieu Lauri\`ere, and Tim Byrnes.
\newblock Multipartite generalization of quantum discord.
\newblock \emph{Phys. Rev. Lett.}, 124:\penalty0 110401, Mar 2020.
\newblock \doi{10.1103/PhysRevLett.124.110401}.
\newblock URL \url{https://link.aps.org/doi/10.1103/PhysRevLett.124.110401}.

\bibitem[Horodecki et~al.(2005)Horodecki, Oppenheim, and Winter]{Horodecki2005}
Micha{\l} Horodecki, Jonathan Oppenheim, and Andreas Winter.
\newblock Partial quantum information.
\newblock \emph{Nature}, 436\penalty0 (7051):\penalty0 673--676, August 2005.
\newblock \doi{10.1038/nature03909}.
\newblock URL \url{https://doi.org/10.1038/nature03909}.

\bibitem[Horodecki et~al.(2006)Horodecki, Oppenheim, and Winter]{Horodecki2006}
Micha{\l} Horodecki, Jonathan Oppenheim, and Andreas Winter.
\newblock Quantum state merging and negative information.
\newblock \emph{Communications in Mathematical Physics}, 269\penalty0
  (1):\penalty0 107--136, October 2006.
\newblock \doi{10.1007/s00220-006-0118-x}.
\newblock URL \url{https://doi.org/10.1007/s00220-006-0118-x}.

\bibitem[Bennett and Wiesner(1992)]{Bennett1992}
Charles~H. Bennett and Stephen~J. Wiesner.
\newblock Communication via one- and two-particle operators on
  einstein-podolsky-rosen states.
\newblock \emph{Phys. Rev. Lett.}, 69:\penalty0 2881--2884, Nov 1992.
\newblock \doi{10.1103/PhysRevLett.69.2881}.
\newblock URL \url{https://link.aps.org/doi/10.1103/PhysRevLett.69.2881}.

\bibitem[Bru\ss{} et~al.(2004)Bru\ss{}, D'Ariano, Lewenstein, Macchiavello,
  Sen(De), and Sen]{Bruss2004}
D.~Bru\ss{}, G.~M. D'Ariano, M.~Lewenstein, C.~Macchiavello, A.~Sen(De), and
  U.~Sen.
\newblock Distributed quantum dense coding.
\newblock \emph{Phys. Rev. Lett.}, 93:\penalty0 210501, Nov 2004.
\newblock \doi{10.1103/PhysRevLett.93.210501}.
\newblock URL \url{https://link.aps.org/doi/10.1103/PhysRevLett.93.210501}.

\bibitem[Prabhu et~al.(2013)Prabhu, Pati, Sen(De), and Sen]{Prabhu2013}
R.~Prabhu, Arun~Kumar Pati, Aditi Sen(De), and Ujjwal Sen.
\newblock Exclusion principle for quantum dense coding.
\newblock \emph{Phys. Rev. A}, 87:\penalty0 052319, May 2013.
\newblock \doi{10.1103/PhysRevA.87.052319}.
\newblock URL \url{https://link.aps.org/doi/10.1103/PhysRevA.87.052319}.

\bibitem[Devetak and Winter(2005)]{Devetak2005}
Igor Devetak and Andreas Winter.
\newblock Distillation of secret key and entanglement from quantum states.
\newblock \emph{Proceedings of the Royal Society A: Mathematical, Physical and
  Engineering Sciences}, 461\penalty0 (2053):\penalty0 207--235, January 2005.
\newblock \doi{10.1098/rspa.2004.1372}.
\newblock URL \url{https://doi.org/10.1098/rspa.2004.1372}.

\bibitem[Yang et~al.(2019)Yang, Horodecki, and Winter]{Yang2019}
Dong Yang, Karol Horodecki, and Andreas Winter.
\newblock Distributed private randomness distillation.
\newblock \emph{Phys. Rev. Lett.}, 123:\penalty0 170501, Oct 2019.
\newblock \doi{10.1103/PhysRevLett.123.170501}.
\newblock URL \url{https://link.aps.org/doi/10.1103/PhysRevLett.123.170501}.

\bibitem[Berta et~al.(2010)Berta, Christandl, Colbeck, Renes, and
  Renner]{Berta2010}
Mario Berta, Matthias Christandl, Roger Colbeck, Joseph~M. Renes, and Renato
  Renner.
\newblock The uncertainty principle in the presence of quantum memory.
\newblock \emph{Nature Physics}, 6\penalty0 (9):\penalty0 659--662, July 2010.
\newblock \doi{10.1038/nphys1734}.
\newblock URL \url{https://doi.org/10.1038/nphys1734}.

\bibitem[Vempati et~al.(2021)Vempati, Ganguly, Chakrabarty, and
  Pati]{Vempati2021}
Mahathi Vempati, Nirman Ganguly, Indranil Chakrabarty, and Arun~K. Pati.
\newblock Witnessing negative conditional entropy.
\newblock \emph{Phys. Rev. A}, 104:\penalty0 012417, Jul 2021.
\newblock \doi{10.1103/PhysRevA.104.012417}.
\newblock URL \url{https://link.aps.org/doi/10.1103/PhysRevA.104.012417}.

\bibitem[Friis et~al.(2017)Friis, Bulusu, and Bertlmann]{Friis2017}
Nicolai Friis, Sridhar Bulusu, and Reinhold~A Bertlmann.
\newblock Geometry of two-qubit states with negative conditional entropy.
\newblock \emph{Journal of Physics A: Mathematical and Theoretical},
  50\penalty0 (12):\penalty0 125301, feb 2017.
\newblock \doi{10.1088/1751-8121/aa5dfd}.
\newblock URL \url{https://doi.org/10.1088/1751-8121/aa5dfd}.

\bibitem[Patro et~al.(2017)Patro, Chakrabarty, and Ganguly]{Patro2017}
Subhasree Patro, Indranil Chakrabarty, and Nirman Ganguly.
\newblock Non-negativity of conditional von neumann entropy and global unitary
  operations.
\newblock \emph{Phys. Rev. A}, 96:\penalty0 062102, Dec 2017.
\newblock \doi{10.1103/PhysRevA.96.062102}.
\newblock URL \url{https://link.aps.org/doi/10.1103/PhysRevA.96.062102}.

\bibitem[Nielsen and Chuang(2009)]{Nielsen2000}
Michael~A. Nielsen and Isaac~L. Chuang.
\newblock \emph{Quantum Computation and Quantum Information}.
\newblock Cambridge University Press, 2009.
\newblock \doi{10.1017/cbo9780511976667}.
\newblock URL \url{https://doi.org/10.1017/cbo9780511976667}.

\bibitem[Ku\ifmmode~\acute{s}\else \'{s}\fi{} and \ifmmode~\dot{Z}\else
  \.{Z}\fi{}yczkowski(2001)]{Kuifmmode2001}
Marek Ku\ifmmode~\acute{s}\else \'{s}\fi{} and Karol \ifmmode~\dot{Z}\else
  \.{Z}\fi{}yczkowski.
\newblock Geometry of entangled states.
\newblock \emph{Phys. Rev. A}, 63:\penalty0 032307, Feb 2001.
\newblock \doi{10.1103/PhysRevA.63.032307}.
\newblock URL \url{https://link.aps.org/doi/10.1103/PhysRevA.63.032307}.

\bibitem[Halder et~al.(2021)Halder, Mal, and Sen(De)]{Halder2021}
Saronath Halder, Shiladitya Mal, and Aditi Sen(De).
\newblock Characterizing the boundary of the set of absolutely separable states
  and their generation via noisy environments.
\newblock \emph{Phys. Rev. A}, 103:\penalty0 052431, May 2021.
\newblock \doi{10.1103/PhysRevA.103.052431}.
\newblock URL \url{https://link.aps.org/doi/10.1103/PhysRevA.103.052431}.

\bibitem[Bennett et~al.(1993)Bennett, Brassard, Cr\'epeau, Jozsa, Peres, and
  Wootters]{Bennett1993}
Charles~H. Bennett, Gilles Brassard, Claude Cr\'epeau, Richard Jozsa, Asher
  Peres, and William~K. Wootters.
\newblock Teleporting an unknown quantum state via dual classical and
  einstein-podolsky-rosen channels.
\newblock \emph{Phys. Rev. Lett.}, 70:\penalty0 1895--1899, Mar 1993.
\newblock \doi{10.1103/PhysRevLett.70.1895}.
\newblock URL \url{https://link.aps.org/doi/10.1103/PhysRevLett.70.1895}.

\bibitem[Mendl and Wolf(2009)]{Mendl2009}
Christian~B. Mendl and Michael~M. Wolf.
\newblock Unital quantum channels {\textendash} convex structure and revivals
  of birkhoff's theorem.
\newblock \emph{Communications in Mathematical Physics}, 289\penalty0
  (3):\penalty0 1057--1086, May 2009.
\newblock \doi{10.1007/s00220-009-0824-2}.
\newblock URL \url{https://doi.org/10.1007/s00220-009-0824-2}.

\bibitem[Rains(1999)]{Rains1999}
E.~M. Rains.
\newblock Bound on distillable entanglement.
\newblock \emph{Phys. Rev. A}, 60:\penalty0 179--184, Jul 1999.
\newblock \doi{10.1103/PhysRevA.60.179}.
\newblock URL \url{https://link.aps.org/doi/10.1103/PhysRevA.60.179}.

\bibitem[Rains(2001)]{Rains2001}
E.M. Rains.
\newblock A semidefinite program for distillable entanglement.
\newblock \emph{IEEE Transactions on Information Theory}, 47\penalty0
  (7):\penalty0 2921--2933, 2001.
\newblock \doi{10.1109/18.959270}.

\bibitem[Regula et~al.(2019)Regula, Fang, Wang, and Gu]{Regula2019}
Bartosz Regula, Kun Fang, Xin Wang, and Mile Gu.
\newblock One-shot entanglement distillation beyond local operations and
  classical communication.
\newblock \emph{New Journal of Physics}, 21\penalty0 (10):\penalty0 103017,
  October 2019.
\newblock \doi{10.1088/1367-2630/ab4732}.
\newblock URL \url{https://doi.org/10.1088/1367-2630/ab4732}.

\bibitem[Jiang et~al.(2013)Jiang, Luo, and Fu]{Jiang2013}
Min Jiang, Shunlong Luo, and Shuangshuang Fu.
\newblock Channel-state duality.
\newblock \emph{Phys. Rev. A}, 87:\penalty0 022310, Feb 2013.
\newblock \doi{10.1103/PhysRevA.87.022310}.
\newblock URL \url{https://link.aps.org/doi/10.1103/PhysRevA.87.022310}.

\bibitem[Nielsen(2002)]{Nielsen2002}
Michael~A. Nielsen.
\newblock An introduction to majorization and its applications to quantum
  mechanics.
\newblock 2002.
\newblock URL \url{https://michaelnielsen.org/blog/talks/2002/maj/book.ps}.

\bibitem[Li and Busch(2013)]{Li2013}
Yuan Li and Paul Busch.
\newblock Von neumann entropy and majorization.
\newblock \emph{Journal of Mathematical Analysis and Applications},
  408\penalty0 (1):\penalty0 384--393, December 2013.
\newblock \doi{10.1016/j.jmaa.2013.06.019}.
\newblock URL \url{https://doi.org/10.1016/j.jmaa.2013.06.019}.

\bibitem[Brandsen et~al.(2021)Brandsen, Geng, Wilde, and Gour]{Brandsen2021}
Sarah Brandsen, Isabelle~J Geng, Mark~M Wilde, and Gilad Gour.
\newblock Quantum conditional entropy from information-theoretic principles.
\newblock \emph{arXiv preprint arXiv:2110.15330}, 2021.
\newblock URL \url{https://arxiv.org/abs/2110.15330}.

\end{thebibliography}
\end{document}